\documentclass[10pt]{article}
\usepackage[utf8]{inputenc}
\usepackage{amsthm}
\newtheorem{theorem}{Theorem}
\newtheorem{lemma}[theorem]{Lemma}
\usepackage{caption}
\author{Matthieu Rosenfeld}
\usepackage{amsfonts,amsmath,amssymb,thmtools,thm-restate}

\usepackage[english]{babel}
\usepackage[T1]{fontenc}
\usepackage{enumerate}
\usepackage[ruled,vlined]{algorithm2e}
\usepackage{graphicx}
\usepackage{hyperref}
 \newtheorem{Problem}[theorem]{Problem}

\begin{document}
\title{How far away must forced letters be so that squares are still avoidable?}
\maketitle
\begin{abstract}
We describe a new non-constructive technique to show that squares are avoidable by an infinite word even if we force some letters from the alphabet to appear at certain occurrences.
We show that as long as forced positions are at distance at least 19 (resp. 3, resp. 2) from each other then we can avoid squares over 3 letters (resp. 4 letters, resp. 6 or more letters).
We can also deduce exponential lower bounds on the number of solutions.
For our main Theorem to be applicable, we need to check the existence of some languages and we explain how to verify that they exist with a computer.
We hope that this technique could be applied to other avoidability questions where the good approach seems to be non-constructive (e.g., the Thue-list coloring number of the infinite path).
\end{abstract}

\section{Introduction}
A square is a word of the form $uu$ where $u$ is a non empty word.
We say that a word is square-free (or avoids squares) if none of its factors is a square.
For instance, $hotshots$ is a square while $minimize$ is square-free.
In 1906, Thue showed that there are arbitrarily long ternary words avoiding squares \cite{Thue06}.
This result is  often regarded as the starting point of combinatorics on words, and the generalizations of this particular question received a lot of attention.
The authors of \cite{mainquestionpandd} study three such questions asked by Harju \cite{HARJU2018}.
They also introduced a stronger version of the third problem.
\begin{Problem}[{\cite[Problem 4]{mainquestionpandd}}]
 Let $p\ge2$ be an integer and let $v=v_1v_2v_3\ldots$ be any infinite ternary word. Does there exist an infinite ternary square-free word $w=w_1w_2w_3\ldots$ such that
for all $i$, $w_{p\cdot i}=v_i$?
\end{Problem}

They give a partial solution to this question and they show that the answer is yes for any $v$ if $p\ge30$.
In fact, they showed something slightly stronger.
Let $d(\Sigma)$ be the smallest integer such that for all $v\in\Sigma^\omega$ and 
for all sequence of indices $(p_i)_{1\le i}$ such that $\forall i$, $p_{i+1}-p_i\ge d(\Sigma)$, there is an infinite square-free
word $u\in\Sigma^\omega$ such that $v_i= u_{p_i}$.
They showed that $6\le d(\{0,1,2\})\le30$. 
Moreover, the fact that squares are avoidable over 3 letters can be used to show that $d(\{0,1,2,3\})\le 7$.
We show that $7\le d(\{0,1,2\})\le19$, $d(\{0,1,2,3\})=3$, $d(\{0,1,2,...,k\})=2$ for $k\ge6$.

The main theorem of this paper gives sufficient conditions for the existence of square-free languages that fulfill some constraints.
Kolpakov showed that there are more than $1.30125^n$ square-free words of length $n$ over a ternary alphabet using a new non-constructive technique \cite{Kolpakov2007}. 
One of the ideas behind Kolpakov's result is roughly to approximate (using a computer) the language of square-free words by 
the language of words avoiding squares of period less than $l$ for large $l$, 
and to show that we do not lose too many words if we remove the larger squares from this language.
We use a similar idea in this paper.
We also use ideas from the power series method (see for instance \cite{BELL20071295,doublepat}) even if we do not explicitly manipulate any power series.
It seems to be a good approach to show that the Thue-list number of paths is $3$ 
(see \cite{CZERWINSKI2007453,doi:10.1002/rsa.20411} for definitions and conjectures on this topic) or to tackle other problems that might require a non-constructive approach.

This paper is organized as follows. 
We start by fixing some notations in Section \ref{secdef}.
In Section \ref{secidea}, we give a weaker version of Theorem \ref{mainTh} to present the ideas of the theorem without some of the technicalities.
Then in Section \ref{secentr}, we give the proof of Theorem \ref{mainTh}, our main theorem. 
In Section \ref{secfindingL}, we explain how to verify with a computer the existence of some languages that are required to apply Theorem \ref{mainTh}.
Finally, in Section \ref{secconc}, we use Theorem \ref{mainTh} to bound the values of $d$ for different alphabet sizes.

\section{Definitions and notations}\label{secdef}
We denote the set of non-negative integer (resp. positive integers) by $\mathbb{N}_0$ (resp. $\mathbb{N}_{>0}$).
For any word $w\in\Sigma^*$, we denote the $i$th letter of $w$ by $w_i$ and the length of $w$ by $|w|$.
Then for any $w\in\Sigma^*$, $w=w_1w_2\ldots w_{|w|}$.
For any set of non-empty words $W$, we let $W^*$ (resp. $W^\omega$) be the set of words obtained by catenation of
finitely  many (resp. infinitely many) elements of $W$.
A \emph{language} over an alphabet is a set of finite words over this alphabet. 
We use the convention that $\prod_{x\in\emptyset}x = 1$ and $\max_{x\in\emptyset}x = 0$ 
(we could use $-\infty$ for the second one, but it is slightly less convenient for the implementation).

A \emph{partial word over $\Sigma$} is a (possibly infinite) word over the alphabet $\Sigma\cup\{\diamond\}$. 
For any partial word $\mu\in (\Sigma\cup\{\diamond\})^*\cup (\Sigma\cup\{\diamond\})^\omega$ and word $v\in\Sigma^*\cup \Sigma^\omega$,
we say that \emph{$v$ is compatible with $\mu$} if $|v|\le|\mu|$ and $\mu_i\not=\diamond \implies \mu_i=v_i$ for all $i$ such that $v_i$ and $\mu_i$ are defined.
We denote by $S(\mu)$ the set of square-free words that are compatible with the partial word $\mu$.

\section{Idea of Theorem \ref{mainTh}}\label{secidea}
The main Theorem of this paper is Theorem \ref{mainTh}. 
The main idea of this theorem is that if a language avoids short squares and is large enough then it contains square-free words of any length.
The statement and proof of this theorem are rather difficult to follow so we give in this section a version of the Theorem for the case where the set $W$ is a singleton $\{w\}$.
We hope that this helps to convey the ideas of the proof of Theorem \ref{mainTh}.
This is in fact really similar to the ideas of \cite{doublepat}, but instead of building the word letter by letter,
we construct it factor by factor. 
For that we fix one size of a factor and look at the number of words whose length corresponds to multiples of this size.
\begin{theorem}\label{easy}
Let $\Sigma$ be an alphabet, $w\in(\Sigma\cup\{\diamond\})^*$ be a finite partial word and $p\ge2|w|$ such that $|w|$ divides $p$.
Suppose that there are $C\in\mathbb{N}_{>0}$ and $L$ a language such that:
\begin{enumerate}[(I)]
 \item $\varepsilon\in L$.
 \item For all $u\in L$, $u$ avoids squares of period less than $p$.
 \item For any $u\in L$ there are at least $C$ different words $v\in\Sigma^{|w|}$ compatible with $w$ such that $uv\in L$.
 \item \label{condeq}There exists $x\in]0,1[$ such that:
\begin{equation*}
C\left(1-\frac{x^{\frac{p}{|w|}-1}|w|^2 }{1-x}\right) \ge x^{-1}
\end{equation*}
\end{enumerate}
 Then $S(w^\omega)$ is infinite.
\end{theorem}
\begin{proof}
Let $\mu=w^\omega$.
Let $L(\mu)$ be a set of words from $L$ that are compatible with $\mu$ such that,
for any $u\in L(\mu)$ of length divisible by $|w|$, there are exactly $C$ different words $v\in\Sigma^*$ compatible with $w$ with $uv\in L(\mu)$. 
Conditions (I),(II) and (III) imply that such a set can be obtained by removing words from $L$.
For all non-negative $i$, let $s_{i}= \left|S(\mu)\cap L(\mu)\cap \{u\in\Sigma^*: |u|=i|w|\}\right|$
be the number of square-free words of $L(\mu)$ of length $i|w|$.

We will show by induction on $i$ that for all positive $i$, $s_{i+1}\ge x^{-1}s_{i}$.
Let $n$ be a  positive integer such that:
\begin{equation}
 \forall 0\le i<n,\, s_{i+1}\ge x^{-1}s_{i}\tag{IH1}\label{IH11}
\end{equation}
 
By definition of $L(\mu)$, for any word $w$ of $S(\mu)\cap L(\mu)$ there are exactly $C$ different factors $v$ of length $|w|$ such that $wv$ is in $L(\mu)$.
Let $F$ be the set of words in $L(\mu)\setminus S(\mu)$ of length $(n+1)|w|$ whose prefix of length $n|w|$ is in $S(\mu)\cap L(\mu)$. 
Then by definition: 
\begin{equation}\label{eqrecs1}
s_{n+1}\ge Cs_{n}-|F|.
\end{equation}

In order to bound $|F|$, let us introduce for all $i<n+1$, $F_i=\{uvvy\in F: |w|(i-1)<|uv|\le i|w|, |y|<|w|\}$.
That is, $F_i$ is the set of words of $F$ that contain a square whose midpoint (the middle of the square) is located between the positions $(i-1)|w|$ and $i|w|$ in the word.
Clearly $|F|\le \sum_{i=1}^{n} |F_i|$, so our next task is to compute bounds on $|F_i|$ for all $i$.
\begin{lemma}\label{sizefi1}
We have the following inequalities:
\begin{itemize}
 \item for all $i> n+1-\frac{p}{|w|}$, $|F_i|=0$,
 \item for all $i \le n+1-\frac{p}{|w|}$, $|F_{i}|\le s_{i}C|w|^2$.
\end{itemize}
\end{lemma}
\begin{proof}
If $i> n+1-\frac{p}{|w|}$, then $(i-1)|w|+p\ge (n+1)|w|$.
Since, $L$ does not contain squares of period less than $p$, $F_i=\emptyset$.

 Now, let $i \le n+1-\frac{p}{|w|}$.
 For any $i$ and $z\in S(\mu)\cap L(\mu)\cap \{u\in\Sigma^*: |u|=i|w|\}$ let $F_i(z)$ be the set of words of $F_i$ that admit $z$ as a prefix.
 
 By definition of $F_i$, any word  $F_i(z)$ contains a square whose second half starts in position $a+1$ and ends in position $b$ where $(i-1)|w|< a\le i|w|$ and $ n|w|< b\le (n+1)|w|$.
 Given $z$, $a$ and $b$, we know the first half of the square and thus the word is known at least up to position $n|w|$.
 By definition of $L(\mu)$ there are at most $C$ possible values for the remaining $|w|$ letters.
 By summing over all the values of $a$ and $b$ one gets: $|F_{i}(z)|\le C|w|^2$.
 By summing over all the values of $z$, we finally get $|F_{i}|\le s_{i}C|w|^2$.
\end{proof}
Now, by \eqref{IH11}, for all $i$, $|F_i|\le C|w|^2x^{n-i}s_n$ and thus:
\begin{align*}
|F|\le& \sum_{i=1}^{n} |F_i|\le \sum_{i=1}^{n+1-\frac{p}{|w|}} C|w|^2x^{n-i}s_n\le C|w|^2s_n \sum_{i=\frac{p}{|w|}-1}^{\infty} x^{i}\\
|F|\le& C|w|^2s_n \frac{x^{\frac{p}{|w|}-1}}{1-x}
\end{align*}

We can use this bound in inequality (\ref{eqrecs1}) and we get:
\begin{align*}
s_{n+1}&\ge Cs_{n}-C|w|^2s_n \frac{x^{\frac{p}{|w|}-1}}{1-x}\\
 s_{n+1}&\ge  s_nC\left(1-\frac{x^{\frac{p}{|w|}-1}|w|^2 }{1-x}\right)\\
 s_{n+1}&\ge  x^{-1}s_n\text{\ \ \ (By Theorem hypothesis \ref{condeq})}
\end{align*}

This concludes the proof that for all positive $i$, $s_{i+1}\ge x^{-1}s_{i}$. 
Since $s_0=1$, we deduce that $s_i$ is unbounded and thus $S(w^\omega)$ is infinite.
\end{proof}

\section{The main theorem}\label{secentr}
This section is devoted to the proof of the main Theorem.
As already mentioned the ideas of the proof are the same as for the proof of Theorem \ref{easy}.
However, this is more technical because $W$ is not a singleton anymore.
Moreover, we want the equivalent of condition \eqref{condeq} to be as general as possible and for that,
we need to bound the size of $|F|$ as tightly as possible. 
Thus the equivalent of Lemma \ref{sizefi1} (Lemma \ref{sizefi2}) is much more technical and we delay its proof to a later subsection.
\begin{theorem}\label{mainTh}
Let $\Sigma$ be an alphabet, $W\subseteq(\Sigma\cup\{\diamond\})^*$ be a finite set of finite partial words,
$p\ge2\max\{|w|: w\in W\}$ be an integer.
Suppose that there is a language $L$ and a function $f:\mathbb{N}_{>0}\rightarrow \mathbb{N}_{>0}$ such that:
\begin{enumerate}[(I)]
\item $\varepsilon\in L$.
\item For any $u\in L$ and $w\in W$ there are at least $f(|w|)$ different words $v\in\Sigma^{|w|}$ compatible with $w$ and such that $uv\in L$.
\item For all $u\in L$, $u$ avoids squares of period less than $p$.
\item \label{condeq2} For all $u,v\in W$ and integer $1\le i\le|v|$, let
$$\alpha(|u|,|v|)=\sum_{m=1}^{|u|} \sum_{j=0}^{\left\lfloor\frac{|v|-1}{m}\right\rfloor}\min\left\{f(|v|),(|\Sigma|-1)^{|v|-1-jm}\right\}$$
and $$\alpha'(i,|v|) =\sum_{m=0}^{i-1} \min\left\{f(|v|),(|\Sigma|-1)^{m}\right\}.$$
There exist $x_1,x_2,\ldots, x_{\max\{|w|: w\in W\}}\in]0,1[$ and $\beta:\{0,\ldots,p\}\rightarrow [0,1]$ solution of the following system:
\begin{equation*}
 \left\{
 \begin{array}{l}
  \forall w\in W,\\ f(|w|) -
 \max\limits_{\substack{u,v\in W \\1\le r\le |v|}}\left\{\beta\left(r+p-|w|-|v|\right)\left(\alpha'(r,|w|)+\frac{x_{|v|}\alpha(|u|,|w|)}{1-x_{|u|}}\right)\right\} \ge x_{|w|}^{-1}\\
  \forall j\le p, \beta(j)=\max\left\{\prod\limits_{i\in\{|u|:u\in W\}} x_i^{n_i} \Bigg|
 \begin{array}{c}
 \forall i\in\{|u|:u\in W\} , n_i\in \mathbb{N}_0,\\
 \text{and}\\
 \sum\limits_{i\in\{|u|:u\in W\}} i \cdot n_i= j
 \end{array}\right\}\\
\end{array}
\right.
\end{equation*}
\end{enumerate}
Then for any infinite partial word $\mu\in W^\omega$, $S(\mu)$ is infinite.
\end{theorem}
\begin{proof}
Let $\mu\in W^\omega$ and $(\mu_i)_{i\in\mathbb{N}_{>0}}\in W^{\mathbb{N}_{>0}}$ be a sequence of elements of $W$ such that $\mu=\mu_1\mu_2\mu_3\ldots$.
For any integer $i$, let $l(i)= |\mu_1\ldots \mu_i|$.
Let $L(\mu)$ be a set of words from $L$ that are compatible with $\mu$ such that,
for any $u\in L(\mu)$ of length $|\mu_1\ldots \mu_j|$, there are exactly $f(|\mu_{j+1}|)$ different words $v\in\Sigma^*$ compatible with $\mu_{j+1}$ with
$uv\in L(\mu)$. That is, we remove words from $L$ in order to replace the ``at least $f(|\mu_{j+1}|)$'' by ``exactly $f(|\mu_{j+1}|)$''.
For all non-negative $i$, let $s_{i}= \left|S(\mu)\cap L(\mu)\cap \{v\in\Sigma^*: |v|=i\}\right|$
be the number of square-free words of $L(\mu)$ of length $i$.

We will show by induction on $i$ that for all positive $i$, $s_{l(i+1)}\ge x_{|\mu_{i+1}|}^{-1}s_{l(i)}$.

Let $n$ be a  positive integer such that:
\begin{equation}
 \forall 0\le i<n,\, s_{l(i+1)}\ge x_{|\mu_{i+1}|}^{-1}s_{l(i)}\tag{IH1}\label{IH1}
\end{equation}
 
By definition of $L(\mu)$, for any word $w$ of $S(\mu)\cap L(\mu)$ of length $l(n)$ there are exactly $f(|\mu_{n+1}|)$ different factors $v$ of length $|\mu_{n+1}|$ such that $wv$ is in $L(\mu)$.
Let $F$ be the set of words in $L(\mu)\setminus S(\mu)$ of length $l(n+1)$ whose prefix of length $l(n)$ is in $S(\mu)\cap L(\mu)$. 
Then by definition: 
\begin{equation}\label{eqrecs}
s_{l(n+1)}\ge f(|\mu_{n+1}|)s_{l(n)}-|F|.
\end{equation}

In order to bound $|F|$, let us introduce for all $i<n+1$, $F_i=\{uvvw\in F: l(i-1)<|uv|\le l(i), |w|<|\mu_{n+1}|\}$.
That is, $F_i$ is the set of words of $F$ that contain a square whose midpoint (the middle of the square) is located between the positions $|\mu_1\ldots \mu_i|$ and $|\mu_1\ldots \mu_{i-1}|$ in the word.
Clearly $|F|\le \sum_{i=1}^{n} |F_i|$, so our next task is to compute the values of $|F_i|$ for all i.
Let $d$ be the smallest integer such that $|\mu_{d+1}\ldots \mu_{n+1}|\le p$ and let $r=|\mu_{d}\mu_{d+1}\ldots \mu_{n+1}|-p$.
Remark that $r>0$.
\begin{restatable}{lemma}{sizefii}
\label{sizefi2}
We have the following inequalities:
\begin{itemize}
 \item for all $i> d$, $|F_i|=0$,
 \item $|F_{d}|\le s_{l(d)}\alpha'(r,|\mu_{n+1}|)$,
 \item for all $i \le d$, $|F_{i}|\le s_{l(i)}\alpha(|\mu_{i}|,|\mu_{n+1}|)$.
\end{itemize}
\end{restatable}
The proof of this Lemma is not really informative and is mostly a rather technical counting argument, so we moved it to Section \ref{secprooflemma}.

We can use the bounds on the sizes of the $F_i$s to bound $|F|$:
\begin{lemma}\label{sizeofF}
We have $$|F|\le s_{l(d)}\max_{u\in W}\left\{\alpha'(r,|\mu_{n+1}|)+\frac{x_{|\mu_d|}\alpha(|u|,|\mu_{n+1}|)}{1-x_{|u|}}\right\}.$$
\end{lemma}
\begin{proof}
First, let us show by induction on $i$ that for all $0\le i<d$: 
\begin{equation*}\sum_{ j=0}^{i} s_{l(j)}\alpha(|\mu_j|,|\mu_{n+1}|)\le s_{l(i)}\max_{u\in W}\left\{\frac{\alpha(|u|,|\mu_{n+1}|)}{1-x_{|u|}}\right\}.\tag{IH2}\label{IH2}
\end{equation*}

Let us first show that this is true with $i=0$, using the fact that $\mu_0\in]0,1[$.
$$s_{l(0)}\alpha(|\mu_0|,|\mu_{n+1}|)\le
s_{l(0)}\frac{\alpha(|\mu_0|,|\mu_{n+1}|)}{1-x_{|\mu_0|}}\le
s_{l(0)}\max_{u\in W}\left\{\frac{\alpha(|u|,|\mu_{n+1}|)}{1-x_{|u|}}\right\}.
$$

Now, let $i+1$ be an integer such that \eqref{IH2} is true for $i$.
\begin{align*}
 &\sum_{ j=0}^{i+1} s_{l(j)}\alpha(|\mu_j|,|\mu_{n+1}|)\le s_{l(i+1)}\alpha(|\mu_{i+1}|,|\mu_{n+1}|)+s_{l(i)}\max_{u\in W}\left\{\frac{\alpha(|u|,|\mu_{n+1}|)}{1-x_{|u|}}\right\}\\
 &\le s_{l(i+1)}\left(\alpha(|\mu_{i+1}|,|\mu_{n+1}|)+x_{|\mu_{i+1}|}\max_{u\in W}\left\{\frac{\alpha(|u|,|\mu_{n+1}|)}{1-x_{|u|}}\right\}\right) \text{ (By \eqref{IH1})}\\
 &\le s_{l(i+1)}\max_{u\in W}\left\{\frac{\alpha(|u|,|\mu_{n+1}|)}{1-x_{|u|}}\right\}
\end{align*}
Thus equation \eqref{IH2} is true for all $i\le d$ and in particular for $i=d-1$ and we get:
\begin{align*}
 |F|&\le|F_d|+ s_{l(d-1)}\max_{u\in W}\left\{\frac{\alpha(|u|,|\mu_{n+1}|)}{1-x_{|u|}}\right\}\\
 |F|&\le s_{l(d)}\alpha'(r,|\mu_{n+1}|)+ s_{l(d)}x_{|\mu_{d}|} \max_{u\in W}\left\{\frac{\alpha(|u|,|\mu_{n+1}|)}{1-x_{|u|}}\right\}\\
 |F|&\le s_{l(d)}\max_{u\in W}\left\{\alpha'(r,|\mu_{n+1}|)+\frac{x_{|\mu_{d}|}\alpha(|u|,|\mu_{n+1}|)}{1-x_{|u|}}\right\}
\end{align*}
This concludes the proof of this Lemma.
\end{proof}

By induction hypothesis \eqref{IH1} $s_{l(d)}\le s_{l(n)}\prod_{i=d+1}^{n}x_{|\mu_{j}|}$.
Let us bound the product on the right hand side:
\begin{align*}
 \prod_{i=d+1}^{n}x_{|\mu_{j}|}&\le\max\left\{\prod\limits_{i\in\{|u|:u\in W\}} x_i^{n_i} \Bigg|
 \begin{array}{c}
 \forall i\in\{|u|:u\in W\} , n_i\in \mathbb{N}_{0}\\
 \text{and}\\
 \sum\limits_{i\in\{|u|:u\in W\}} i \cdot n_i= l(n+1)-l(d)-\mu_{n+1}
 \end{array}\right\}\\
  \prod_{i=d+1}^{n}x_{|\mu_{j}|}&\le\beta\left(l(n+1)-l(d)-\mu_{n+1}\right)\\
  \prod_{i=d+1}^{n}x_{|\mu_{j}|}&\le\beta\left(r+p-\mu_{n+1}-\mu_{d}\right)
\end{align*}
Now, using this equation  with Lemma \ref{sizeofF} gives 
$$|F|\le s_{l(n)}\beta\left(r+p-\mu_{n+1}-\mu_{d}\right)\max_{u\in W}\left\{\alpha'(r,|\mu_{n+1}|)+\frac{x_{|\mu_{d}|}\alpha(|u|,|\mu_{n+1}|)}{1-x_{|u|}}\right\}$$
Now recall that $r=|\mu_{d}\mu_{d+1}\ldots \mu_{n+1}|-p$ and thus by definition of $d$, $1\le r\le\mu_d$.
We deduce:
$$|F|\le s_{l(n)}\max\limits_{\substack{u,v\in W \\1\le r\le |v|}}\left\{\beta\left(r+p-\mu_{n+1}-|v|\right)\left(\alpha'(r,|\mu_{n+1}|)+\frac{x_{|v|}\alpha(|u|,|\mu_{n+1}|)}{1-x_{|u|}}\right)\right\}$$

We can finally replace $|F|$ by this bound in inequality (\ref{eqrecs}) and we get:
\begin{align*}
 &s_{l(n+1)}\ge  s_{l(n)}\Bigg(f(|\mu_{n+1}|) -\\
 &\max\limits_{\substack{u,v\in W \\r\in\{1,\ldots, |v|\}}}\left\{\beta\left(r+p-\mu_{n+1}-|v|\right)\left(\alpha'(r,|\mu_{n+1}|)+\frac{x_{|v|}\alpha(|u|,|\mu_{n+1}|)}{1-x_{|u|}}\right)\right\}\Bigg)\\
 &s_{l(n+1)}\ge  s_{l(n)}x_{|\mu_{n+1}|}^{-1} \text{\ \ \ (By Theorem hypothesis \eqref{condeq2})}
\end{align*}
Moreover $s_0=1$ and thus for all $i$, $s_{|\mu_1\ldots\mu_i|}\ge\prod_{j=1}^i x_{|\mu_j|}^{-1}$.
For all $j$, $x_{|\mu_j|}^{-1}>1$, so we conclude that  $S(\mu)$ is infinite.
\end{proof}
Remark that Theorem \ref{mainTh} is far from sharp. 
One could improve the bounds given by Lemma \ref{sizefi2}.
This could be done by lowering $\alpha$ and $\alpha'$ or by introducing a third coefficient $\alpha''$ for the second non-empty $F_i$.
However, we were not able to obtain significant improvement that were worth the additional technicalities.

In Section \ref{secfindingL} we explain how to verify with a computer that there exists a language $L$ that satistfies conditions (I),(II) and (III).
We also need a way to verify condition (IV). 
In order to compute $\beta$, we can use that $\beta(0)=1$ and, for all $j\in \left\{1,\ldots, p\right\}$, $\beta(j)=\max\left\{x_{|u|}\beta(j-|u|): u\in W, |u|\le j\right\}$.
Thus given the values of the $x_i$ one can compute $\beta$ using a dynamic algorithm and all the rest is straight forward to compute. 
Thus it is easy to verify with a computer whether or not a given set of values of $x_i$ is a solution.
We provide a C++ program that takes as input $|\Sigma|$, $k$, $p$, $f$ and $x_1,\ldots, x_k$ and
verifies whether this is a solution of the equations of condition (IV).

\subsection{Proof of Lemma \ref{sizefi2}}\label{secprooflemma}
This subsection is dedicated to the proof of Lemma \ref{sizefi2}.
Remark that the statement and proof are not self-contained since some of the notations are defined in the proof of Theorem \ref{mainTh}.

\sizefii*

\begin{proof}
 If $i> d$ then by definition $|\mu_{i}\ldots \mu_{n+1}|\le p$. 
 Moreover, $L$ does not contain squares of period less than $p$ and thus $F_i=\emptyset$.

 Now, let $i\le d$.
 By definition, any word from $F_i$ can be written $uvvy$ with $l(i-1)<|uv|\le l(i), |y|<|\mu_{n+1}|$. 
 For any $i$ and $z\in S(\mu)\cap L(\mu)\cap \{u\in\Sigma^*: |u|=l(i)\}$, let $F_i(z)$ be the set of  words of $F_i$ that admit $z$ as a prefix.
 Clearly $F_i=\sum_{z\in S(\mu)\cap L(\mu)\cap \{u\in\Sigma^*: |u|=l(i)\}}F_i(z)$.
 
 Let $a,b$ (resp. $a',b'$) be integers such that there is an element $z'\in F_i(z)$ that contains a square starting at $a$ (resp. $a'$)
 and of period $b$ (resp. $b'$) with $l(i-1)+1<a'+b'=a+b\le |z|+1$ and $a>a'$.
 Because of the square in $z'$, we know that for all $0\le j\le |z|-a-b$, $z_{a+j}=z_{a+b+j}=z_{a'+b'+j}=z_{a'+j}$.
 If $a\le a'+|z|-a-b+1$ then $z$ contains a square which is not possible.
 Hence 
 \begin{align}
a&> a'+|z|-a-b+1\nonumber\\
a+b+b'&> a'+b+b'+|z|-a-b+1\nonumber\\
a'+2b'&> a+2b+|z|-a-b+1\ \text{ since }a+b=a'+b'\label{squaresamemiddle}
\end{align}
Let $u\in F_i(z)$ be a word that contains a square starting at $a$ and of period $b$ then we know its suffix of size 
$a+2b-1>l(n)$.
Thus there are at most $f(|\mu_{n+1}|)$ possibilities, moreover since the size of the unknown suffix is $l(n+1)+1-a-2b<p$, it is square-free
and there are at most $(|\Sigma|-1)^{l(n+1)+1-a-2b}$ possibilities.
Thus for a fixed $z$ and value of $a+b$, the number of ways to add a suffix to $z$ to obtain an element of $F_i$ that contains a square of period $b$ starting at $a$ is:
\begin{align*}
 &\max_{ \substack{j\in \mathbb{N}_0,\, l(n)+2\le p_1,\ldots, p_j\le l(n+1)+1,\\
 \forall i<j,\, p_{i+1}>p_{i}+|z|-a-b+1 }}
 \left\{\sum_{i=1}^{j}\min\left\{f(|\mu_{n+1}|),(|\Sigma|-1)^{l(n+1)+1-p_i}\right\} 
\right\}\\
&= \max_ {\substack{j\in \mathbb{N}_0,\, 0\le p_1,\ldots, p_j\le |\mu_{n+1}|-1,\\
 \forall i<j,\, p_{i+1}<p_{i}+|z|-a-b+1}}
 \left\{\sum_{i=1}^{j}\min\left\{f(|\mu_{n+1}|),(|\Sigma|-1)^{|\mu_{n+1}|-1-p_i}\right\}\right\}
\end{align*}
But since $\min\left\{f(|\mu_{n+1}|),(|\Sigma|-1)^{|\mu_{n+1}|-1-p_i}\right\}$ is a non-increasing function in $p_i$, the maximum is reached
when we pack the $p_i$ as much as we can on the lowest values of $p_i$. Thus this quantity is in fact equal to:
\begin{equation*}
\sum_{j=0}^{\left\lfloor\frac{|\mu_{n+1}|-1}{|z|-a-b+2}\right\rfloor}\min\left\{f(|\mu_{n+1}|),(|\Sigma|-1)^{|\mu_{n+1}|-1-j(|z|-a-b+2)}\right\}
\end{equation*}

Now we can sum over all the values of $a+b$ and we get:
\begin{align}
|F_i(z)|&\le \sum_{a+b=|z|+2-|\mu_i|}^{|z|+1} \sum_{j=0}^{\left\lfloor\frac{|\mu_{n+1}|-1}{|z|-a-b+2}\right\rfloor}\min\left\{f(|\mu_{n+1}|),(|\Sigma|-1)^{|\mu_{n+1}|-1-j(|z|-a-b+2)}\right\}\nonumber\\
&\le \sum_{a+b=0}^{|\mu_{i}|-1} \left(\sum_{j=0}^{\left\lfloor\frac{|\mu_{n+1}|-1}{|\mu_i|-a-b}\right\rfloor}\min\left\{f(|\mu_{n+1}|),(|\Sigma|-1)^{|\mu_{n+1}|-1-j(|\mu_i|-a-b)}\right\}\right)\nonumber\\
&\le \sum_{m=1}^{|\mu_{i}|}\left( \sum_{j=0}^{\left\lfloor\frac{|\mu_{n+1}|-1}{m}\right\rfloor}\min\left\{f(|\mu_{n+1}|),(|\Sigma|-1)^{|\mu_{n+1}|-1-jm}\right\}\right)\nonumber\\
|F_i(z)|&\le \alpha(|\mu_{i}|,|\mu_{n+1}|)\nonumber
\end{align}
Summing over all the possible $z$ yields $|F_i|\le s_{l(i)}\alpha(|\mu_{i}|,|\mu_{n+1}|)$.

The remaining case is $i=d$.
Once again, let $a,b$ (resp. $a',b'$) be integers such that there is an element of $F_d(z)$ that contains a square starting at $a$ (resp. $a'$)
 and of period $b$ (resp. $b'$) with $l(d-1)+1<a'+b'=a+b\le  l(n+1)+1-p$ and $a>a'$.
By definition of $d$, $ l(n+1)- l(d)\le p$. 
We can use equation \eqref{squaresamemiddle} again and we get:
$$a'+2b'> b+|z|+1\ge p+|z|+1 \ge l(n+1)+1$$
This is a contradiction with the fact that $a'+2b'\le l(n+1)+1$.
Thus given the value of $a+b$ there is at most one possible value for $a$ and $b$.
The number of ways for a fixed $z$ and value of $a+b$ to complete $z$ with a suffix into an element of $F_d$ is at most:
\begin{align*}
 &\max\left\{\min\left\{f(|\mu_{n+1}|),(|\Sigma|-1)^{l(n+1)+1-s}\right\}: 
 \begin{array}{l} s\in\mathbb{N}_{>0},s\le l(n+1)+1,\\
 a+b+p\le s,\\
 \end{array}\right\}\\
 &\le \min\left\{f(|\mu_{n+1}|),(|\Sigma|-1)^{l(n+1)+1-a-b-p}\right\}
\end{align*}
Then by summing over all the possible values of $a+b$, we get:
$$ |F_d(z)|\le \sum_{a+b=l(d-1)+2}^{l(n+1)+1-p} \min\left\{f(|\mu_{n+1}|),(|\Sigma|-1)^{l(n+1)+1-a-b-p}\right\}$$
We can use the variable substitution $m=l(n+1)+1-a-b-p$ and remark that $l(n+1)+1-p-(l(d-1)+2)=|\mu_{d}\mu_{d+1}\ldots \mu_{n+1}|-p-1=r-1$ and we get:
$$ |F_d(z)|\le \sum_{m=0}^{r-1} \min\left\{f(|\mu_{n+1}|),(|\Sigma|-1)^{m}\right\}\le \alpha'(r,|\mu_{n+1}|)$$
We conclude that $|F_{d}|\le s_{l(d)}\alpha'(r,|\mu_{n+1}|)$ by summing over all $z$.
\end{proof}

\section{Finding a set \texorpdfstring{$L$}{L} that satisfies Theorem \ref{mainTh}}\label{secfindingL}
In this section, we explain how to verify the existence of a language that fulfills conditions (I),(II) and (III) of Theorem \ref{mainTh}.

We consider some particular directed labeled graphs: $G(V,A)$ is a set $V$ of vertices together with a set $A\subseteq(V\times V\times\Sigma)$ of labeled arcs.
For any $u,v\in V$ and $a\in\Sigma$, $(u,v,a)\in A$ is an arc from $u$ to $v$ with label $a$.
These graphs could also be seen as finite state machines where all the states are initial and final.

The Rauzy graph of length $n$ of a factorial language $L$ over $\Sigma$ is the graph $G(V, A)$ where $V = L\cap \Sigma^n$ and
$E=\{(au,ub, b):  aub\in L, a,b\in\Sigma \}$. 
For any graph $G(V,A)$ and any set $X\in V$, we denote by $G[X]$ the subgraph induced by $X$.

Let $R_p(\Sigma)$ be the Rauzy graph of length $2p-3$ of the square-free words over $\Sigma$.
Remark, that the factors of length $2p-2$ of any walk on $R_p(\Sigma)$ correspond to edges of $R_p(\Sigma)$ and by definition they are square free.
Thus, the sequence of labels of any walk on $R_p(\Sigma)$ avoids squares of period less than $p$, but can contain longer squares. 
We let $S_p(\Sigma)$ be the set of words that contains no square of period less than $p$ (from the previous remark $S_p(\Sigma)$ can also be seen as the set of walks on $R_p(\Sigma)$).

As an illustration, we give $R_3(\{0,1,2\})$ in Fig. \ref{fullrauzy} without the arc labels.

\begin{figure}[ht]
\begin{center}
\includegraphics[scale=0.8]{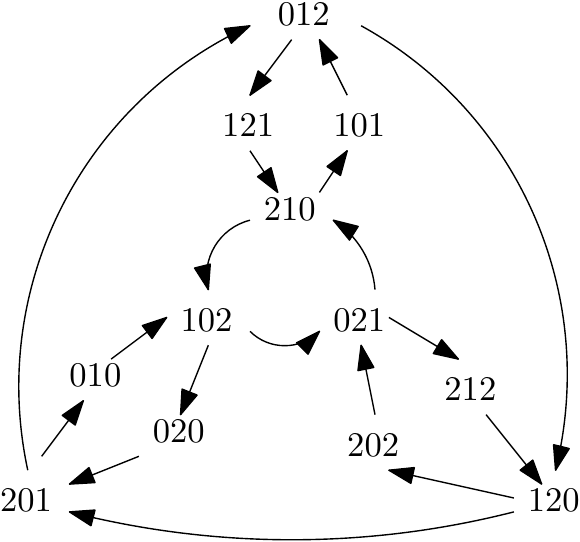}
\end{center}
\caption{The Rauzy graph $R_3(\{0,1,2\})$.\label{fullrauzy}}
\end{figure}

For this Section, we abuse the notation and allow ourself to identify words and sequences.

For any graph $G(V,A)$ and partial word $w\in(\Sigma\cup\{\diamond\})^*$,
we define inductively for any integer $i\in\{0,\ldots, |w|\}$ and vertex $v\in V$: 
$$p_{i,w,G}(v)=\left\{\begin{array}{ccl} 
 1 &&\text{if }i=0,\\  
 \sum\limits_{\substack{(v,u,a)\in A\\a\in \Sigma}}p_{i-1,w,G}(u) &&\text{if }w_{|w|+1-i}=\diamond,\\
 \sum\limits_{\substack{(v,u,w_{|w|+1-i})\in A}}p_{i-1,w,G}(u) &&\text{otherwise.}            
\end{array}\right.
$$
Intuitively, $p_{i,w,G}(v)$ gives the number of walks of length $i$ starting from $v$
that are compatible with the $i$ last letters of $w$.
Indeed, there is one walk of length $0$ and we always take the transition that is labeled by the current letter of $w$ and any transition if this letter is $\diamond$.
Remark that in the third case, there are in fact either $0$ or $1$ summands in the sum.
\begin{lemma}\label{firstRauzy}
Let $W\subseteq(\Sigma\cup\{\diamond\})^*$, $G(V,A)=R_p(\Sigma)$, $f:\mathbb{N}_{>0}\rightarrow\mathbb{N}_{>0}$ and a non-empty set $X\subseteq V$.
If  for all $v\in X$ and $w\in W$, $p_{|w|,w,G[X]}(v)\ge f(|w|)$,
then there exists a language $L$ such that:
 \begin{itemize}
 \item $\varepsilon\in L$,
  \item for all $u\in L$, $u$ avoids squares of period less than $p$,
  \item for any $u\in L$ and $w\in W$ there are at least $f(|w|)$ different words $v\in\Sigma^{|w|}$ compatible with $w$ such that $uv\in L$.
 \end{itemize}
\end{lemma}
\begin{proof}
Let $L$ be the set of sequences of labels that correspond to a walk in $G[X]$.
By definition, the two first conditions on $L$ are fulfilled.
 
Let  $u\in L$ and $w\in W$.
If $|u|\ge2p-3$ we let $u'$ be the suffix of length $2p-3$ of $u$.
Otherwise, we let $u'\in L$ such that  $|u'|=2p-3$ and $w$ is a suffix of $u'$ (there is such an element in $L$).
Each walk of length $|w|$ starting in $u'$ gives a unique sequence of labels $u''$ such that $uu''$ contains no square of period less than $p$. 

We easily deduce, by induction on $i$, that for all $v$ the number of walks of length $i$ starting at $v$ that are compatible with $w_{|w|-i+1}w_{|w|-i+2}\ldots w_{|w|}$
is at least $p_{i,w,G[X]}(v)$.
So, in particular, the number of walks of length $|w|$ starting at $u'$ and compatible with $w$ is at least $p_{|w|,w,G[X]}(u')\ge f(|w|)$.
This concludes the proof.
\end{proof}

In fact, we need something stronger because for the values of $p$ that we use the graphs $R_p$ are too big to fit in a computer.
We can exploit symmetries of $R_p(\Sigma)$ to work on a smaller equivalent graph.

For any square-free word $w\in\Sigma^*$, we let $\Psi(w)$ be the shortest suffix of $w$ such that for all $i\in\{\lceil\frac{|\Psi(w)|}{2}\rceil+1,\ldots,p-1\}$ there exists $k\in\{0,1,\ldots, |\Psi(w)|-i-1\}$ with $\Psi(w)_{|\Psi(w)|-k}\not=\Psi(w)_{|\Psi(w)|-k-i}$.
If $|w|=2p-3$, then $w$ is such a suffix of itself (since $\{\lceil\frac{|w|}{2}\rceil+1,\ldots,p-1\}$ is empty) and thus there is always a shortest suffix.

For instance, with $p=5$ we have $\Psi(0210120)=10120$.  For any letter $\alpha$, the word $0210120\alpha$ is square free if and only if  $10120\alpha$ is square free. Indeed, a new square is necessarily a suffix of $0210120\alpha$, it is enough to look at the two letters in bold in $\mathbf{1}012\mathbf{0}\alpha$ to deduce that there is no square of length $4$ and all the other possible suffix of even length of $0210120\alpha$ are also suffixes of $10120\alpha$.
In fact, for any $w$ of size $2p-3$, the word $wa$ avoids squares if and only if $\Psi(w)a$ avoids squares  (this is proven in the next lemma) and this is the main motivation behind the definition of $\Psi$ (in particular, words with the same image by $\Psi$ can be extended in the same way).

For any graph $G(V,A)$, let $\Psi(G)(\Psi(V),A')$ be the graph such that $A'=\Psi(A)=\{(\Psi(a),\Psi(b),c): (a,b,c)\in A\}$. The next lemma tells us that we only need to consider the walks on $\Psi(G)$ instead of the walks on $G$.

\begin{lemma}\label{lemmaforf}
Let $p$ be a positive integer and $w\in(\Sigma\cup\{\diamond\})^*$,
$G(V,A)=R_p(\Sigma)$ and $X\subseteq \Psi(V)$.
Let $\Psi^{-1}(X)=\{x\in V: \Psi(x)\in X\}$.
Then for all $v\in \Psi^{-1}(X)$, $p_{|w|,w,G[\Psi^{-1}(X)]}(v)= p_{|w|,w,\Psi(G)[X]}(\Psi(v))$.
\end{lemma}
\begin{proof}
Let us first show that for any $a\in\Sigma$ and $v\in S_p(\Sigma)$ with $|v|=2p-3$ if $\Psi(v)a\in S_p(\Sigma)$ then $va\in S_p(\Sigma)$.
Let us show that under these assumptions, for any $i$, $va$ avoids squares of period $i$.
Since $v$ is square free, we only need to show that no suffix of $va$ is a square.
We have to distinguish between two cases:
\begin{itemize}
 \item  $2i\le|\Psi(v)|+1$. Suppose for the sake of contradiction that there is a square of period $i$ in $va$. We deduce that the suffix of length $|\Psi(v)|+1$ of $va$ contains a square of period $i$. That is, $\Psi(v)a$ contains a square of period $i$ which is a contradiction.
 \item $2i\ge|\Psi(v)|+2$. Since $i$ is an integer we get $i\ge\lceil\frac{|\Psi(v)|}{2}\rceil+1$.
 Moreover $|va|=2p-2$ and thus $i\le p-1$. Thus by definition of $\Psi(v)$, there exists
 $k\in\{0,1,\ldots,  |\Psi(v)|-i-1\}$ such that $\Psi(v)_{|\Psi(v)|-k}\not=\Psi(v)_{|\Psi(v)|-k-i}$.
 Thus there is $k\in\{1,\ldots,  |\Psi(v)a|-i-1\}$ such that $(\Psi(v)a)_{|\Psi(v)a|-k}\not=(\Psi(v)a)_{|\Psi(v)a|-k-i}$. Remark, that $|\Psi(v)a|-i-1=|\Psi(v)|-i\le 2i-2-i= i-2$. We conclude that there is $k\in\{0,\ldots, i-2\}$ such that $(va)_{|va|-k}\not=(va)_{|va|-k-i}$.
 This implies that the suffix of length $2i$ of $va$ is not a square of period $i$.
\end{itemize}
We deduce that for any $a\in\Sigma$ and $v\in S_p(\Sigma)$ with $|v|=2p-3$ if $\Psi(v)a\in S_p(\Sigma)$ then $va\in S_p(\Sigma)$.

Let $u\in V$, $v\in\Psi(V)$ and $a\in\Sigma$ such that $(\Psi(u),v,a)\in \Psi(A)$.
By definition of $\Psi(A)$ this implies that there is $(u',v',a)\in A$ with $\Psi(u')=\Psi(u)$ and $\Psi(v')=v$.
Thus $u'a\in S_p(\Sigma)$ and  $\Psi(u)a=\Psi(u')a\in S_p(\Sigma)$.
From the previous paragraph, it implies that $ua$ is square-free. 
Let us show that $\Psi(u_2u_3\ldots u_{|u|}a)=v$.
By definition,
for all $i\in\{\lceil\frac{|\Psi(u')|}{2}\rceil+1,\ldots,p-1\}$ there exists $k\in\{0,1,\ldots, |\Psi(u')|-i-1\}$ with $\Psi(u')_{|\Psi(u')|-k}\not=\Psi(u')_{|\Psi(u')|-k-i}$. 
We easily deduce that
for all $i\in\{\lceil\frac{|\Psi(u')a|}{2}\rceil+1,\ldots,p-1\}$ there exists $k\in\{0,1,\ldots, |\Psi(u'a)|-i-1\}$ with $\Psi(u'a)_{|\Psi(u'a)|-k}\not=\Psi(u'a)_{|\Psi(u'a)|-k-i}$. 
This implies that $v=\Psi(v')$ is a suffix of $\Psi(u')a$. Since $\Psi(u')a= \Psi(u)a$, we deduce that $v$ is also a suffix of $u_2u_3\ldots u_{|u|}a$. Since $v=\Psi(v')$ and $v$ is a suffix of $u_2u_3\ldots u_{|u|}a$, we get that $\Psi(u_2u_3\ldots u_{|u|}a)=v$. 
We showed that if there are $u\in V$, $v\in\Psi(V)$ and $a\in\Sigma$ such that $(\Psi(u),v,a)\in \Psi(A)$, then there  exists $v''$ such that  $(u,v'',a)\in A$ and $\Psi(v'')=v$.

We deduce that for all $u\in V$, $\{(\Psi(u),\Psi(v),a)\in \Psi(A)\}\subseteq \{(\Psi(u),\Psi(v),a): (u,v,a)\in A\}$.
The other inclusion is clear from the definition of $\Psi(A)$ and we get for all $u\in V$ :
\begin{equation}\label{eqfA}
 \{(\Psi(u),\Psi(v),a)\in \Psi(A)\}= \{(\Psi(u),\Psi(v),a): (u,v,a)\in A\}
\end{equation}
By definition of $R_p(\Sigma)$, for any $u$ there is at most one outgoing arc for every label in the set of the right. Since the two sets are equals, we deduce that every vertex of the graph $\Psi(G)$ has at most one outgoing arc for any label. 
Intuitively, \eqref{eqfA} implies (by induction on the length of the walk) that for any $u,v\in V$ the set of labeled walks from $u$ to $v$ in $G$  is equal to the set of labeled walks from $\Psi(u)$ to $\Psi(v)$ in $\Psi(G)$.
We are now ready to show by induction on $i$ that for all $i\in\{0,\ldots, |w|\}$ and $v\in \Psi^{-1}(X)$, $p_{i,w,G[\Psi^{-1}(X)]}(v)= p_{i,w,\Psi(G)[X]}(\Psi(v))$.
By definition $p_{0,w,G[\Psi^{-1}(X)]}(v)=1= p_{0,w,\Psi(G)[X]}(\Psi(v))$. 

Let $n$ be a positive integer such that for all $v\in\Psi^{-1}(X)$,
 \begin{equation}\tag{IH}\label{IHend}
  \forall i<n,\, p_{i,w,G[\Psi^{-1}(X)]}(v)= p_{i,w,\Psi(G)[X]}(\Psi(v)).
 \end{equation}
Then, for all $v\in \Psi^{-1}(X)$, if $w_{|w|+1-i}\not=\diamond$ we get:
\begin{align*}
p_{i,w,G[\Psi^{-1}(X)]}(v)&=\sum\limits_{\substack{(v,u,w_{|w|+1-i})\in A\\u\in \Psi^{-1}(X)}}p_{i-1,w,G[\Psi^{-1}(X)]}(u)\\
p_{i,w,G[\Psi^{-1}(X)]}(v)&=\sum\limits_{\substack{(v,u,w_{|w|+1-i})\in A\\\Psi(u)\in X}}p_{i-1,w,\Psi(G)[X]}(\Psi(u))\text{\ \ \ (From \eqref{IHend})}\\
p_{i,w,G[\Psi^{-1}(X)]}(v)&=\sum\limits_{\substack{(\Psi(v),\Psi(u),w_{|w|+1-i})\in \Psi(A)\\\Psi(u)\in X}}p_{i-1,w,\Psi(G)[X]}(\Psi(u)) \text{\ \ \ (From \eqref{eqfA})}\\
p_{i,w,G[\Psi^{-1}(X)]}(v)&=\sum\limits_{\substack{(\Psi(v),u,w_{|w|+1-i})\in \Psi(A)\\u\in X}}p_{i-1,w,\Psi(G)[X]}(u)\\
p_{i,w,G[\Psi^{-1}(X)]}(v)&=p_{i,w,\Psi(G)[X]}(\Psi(v))
\end{align*}
The case where $w_{|w|+1-i}=\diamond$ is similar.
\end{proof}

Using Lemma \ref{firstRauzy} together with Lemma \ref{lemmaforf}, we get the following lemma:
\begin{lemma}\label{lemmagraphf}
Let $W\subseteq(\Sigma\cup\{\diamond\})^*$, $G(V,A)=R_p(\Sigma)$, $f:\mathbb{N}_{>0}\rightarrow\mathbb{N}_{>0}$ and $X\subseteq \Psi(V)$ be a non-empty set.
If  for all $v\in X$ and $w\in W$, $p_{|w|,w,\Psi(G)[X]}(v)\ge f(|w|)$,
then there exists a language $L$ such that:
 \begin{itemize}
 \item $\varepsilon\in L$,
  \item for all $u\in L$, $u$ avoids squares of period less than $p$,
  \item for any $u\in L$ and $w\in W$ there are at least $f(|w|)$ different words $v\in\Sigma^{|w|}$ compatible with $w$ and such that
  $uv\in L$.
 \end{itemize}
\end{lemma}

The graph $\Psi(R_p(\Sigma))$ is much smaller than $R_p(\Sigma)$ and we can use a computer to check the conditions of this lemma for the values of $p$ that we used.
One should first find the graph $\Psi(R_p(\Sigma))$. 
The following fact allows us to easily compute the set of vertices of $\Psi(R_p(\Sigma))$  without computing $R_p(\Sigma)$:
\begin{lemma}\label{compgfr}
Let $w\in S_p(\Sigma)$.
Then $w \in \Psi(S_p(\Sigma)\cap\Sigma^{2p-3})$ if and only if $w$ is the smallest non-empty suffix of $w$
such that $\Psi(w)=w$.
\end{lemma}

Moreover, given a graph $G$, the definition of $p_{|w|,w,G}$ gives a trivial dynamic algorithm that computes $p_{|w|,w,G}$ in time $O(|\Sigma|\cdot |w|\cdot|G|)$.
Starting with $X=\Psi(R_p(\Sigma))$  and inductively removing from $X$ all the vertices for which $p_{|w|,w,\Psi(R_p(\Sigma))[X]}< f(|w|)$  gives the largest subgraph that meets the conditions of
Lemma \ref{lemmagraphf}. As long as this subgraph is not empty one can then apply Lemma \ref{lemmagraphf}.
Algorithm \ref{algosubgraph} computes the largest subgraph of $\Psi(G)$ with the required property.
\begin{algorithm}[ht]
\SetKwInOut{Input}{Input}
\SetKwInOut{Output}{Output}
\Input{The graph $\Psi(G)$, the set $W$}
\Output{The largest set $X\subseteq \Psi(V)$ such that for all $v\in X$ and $w\in W$, $p_{|w|,w,\Psi(G)[X]}(v)\ge f(|w|)$}

$X=\Psi(V)$\;
$todo := true$\;
\While{todo}{
$todo := false$\;
\ForEach{$w\in W$}{
compute $p_{|w|,w,\Psi(G)[X]}$\;
$X':=\{v\in X: p_{|w|,w,\Psi(G)[X]}(v)\ge f(|w|)\}$\;
\If{$X\not=X'$}{
$X:=X'$\;
$todo := true$\;
}
}
}
\KwRet $X$\;
\caption{How to compute the subgraph of $\Psi(G)$.\label{algosubgraph}}
\end{algorithm}

\section{Application of Theorem \ref{mainTh}}\label{secconc}
In this section we apply Theorem \ref{mainTh}.
We provide a C++ implementation of Algorithm \ref{algosubgraph} that verifies the existence of the language that fulfill conditions (I),(II) and (III) from Theorem \ref{mainTh}.
Condition (IV) can be easily verified (as long as solutions are given) and we also provide a C++ code to do that.
\footnote{The codes can be found in the ancillary files of \url{https://arxiv.org/abs/1903.04214}} 
\begin{theorem}
For any alphabet $\Sigma$, let $d(\Sigma)$ be the smallest integer such that for all $v\in\Sigma^\omega$ and 
for all sequence $(p_i)_{1\le i}$ such that $\forall i$, $p_{i+1}-p_i\ge d(\Sigma)$, there is an infinite square-free
word $u\in\Sigma^\omega$ such that $v_i= u_{p_i}$. Then:
\begin{itemize}
 \item $7\le d(\{0,1,2\})\le19$,
 \item $d(\{0,1,2,3\})=3$,
 \item $2\le d(\{0,1,2,3,4\})\le3$,
 \item if $|\Sigma|\ge6$, $d(\Sigma)=2$.
\end{itemize}
\end{theorem}
\begin{proof}
For any alphabet $\Sigma$, we have $d(\Sigma)\ge2$.
Moreover, $d$ is a decreasing function of the size of the alphabet. Thus the third statement can easily be deduced from the second one.
We show the remaining statements independently of each others. 
We will show the upper bounds using Algorithm \ref{algosubgraph}, Lemma \ref{lemmagraphf} and Theorem \ref{mainTh}. 
The lower bounds are verified by exhaustive search.\\

\item If $\mathbf{|\Sigma|\ge6,\, d(\Sigma)=2}$:
Let $\Sigma = \{0,1,2,3,4,5\}$ and  $W=\{{\diamond}\}\cup\{{\diamond} a{\diamond} : a\in\Sigma\}\cup\{{\diamond} a{\diamond} b: a,b\in\Sigma\}.$
We can use Algorithm \ref{algosubgraph} to check that we can apply Lemma \ref{lemmagraphf} with  $f(1)=3$, $f(3)=f(4)=6$, $p=12$.
Thus conditions (I),(II) and (III) of Theorem \ref{mainTh} are fullfilled.
We can check with a computer that condition  \eqref{condeq2} of Theorem \ref{mainTh} is also fullfilled with $x_1=\frac{2}{5}$ and $x_3=x_4=\frac{1}{4}$.
This implies that for any $\mu\in W^\omega$ there are infinite square-free words over $\Sigma$ compatible with $\mu$.
Since $\{{\diamond} ^i a : i\ge1, a\in \Sigma \}^\omega\subseteq W^\omega$, we deduce that
for any $\mu\in \{{\diamond} ^i a : i\ge1, a\in \Sigma \}^\omega$ there are infinite square-free words over $\Sigma$ compatible with $\mu$.
We get $d(\{0,1,2,4,5,6\})\le2$.\\

\item $\mathbf{d(\{0,1,2,3\})=3}$:
Let $w=(0{\diamond}1{\diamond}2{\diamond}3{\diamond})^\omega$. 
An exhaustive search confirms that there are only $636$ square-free words over $\{0,1,2,3\}$ compatible with $w$.
Thus $d(\{0,1,2,3\})\ge3$.

Let $\Sigma = \{0,1,2,3\}$ and $W=\{{\diamond}\}\cup\{{\diamond}{\diamond} a{\diamond}: a\in\Sigma\}
\cup\{{\diamond}{\diamond}a{\diamond}{\diamond}b: a,b\in\Sigma\}.$
We can use Algorithm \ref{algosubgraph} to check that we can apply Lemma \ref{lemmagraphf} with  $f(1)=2$, $f(4)=5$ and $f(6)=8$, $p=18$.
We can then apply Theorem \ref{mainTh} with $x_1=\frac{11}{20}$, $x_4=\frac{1}{4}$ and $x_6=\frac{1}{5}$ and we deduce that
for any $\mu\in W^\omega$ there are infinite square-free words over $\Sigma$ compatible with $\mu$.
Moreover, $\{\diamond ^i a : i\ge2, a\in \Sigma \}^\omega\subseteq W^\omega.$
We deduce that for any $\mu\in \{\diamond ^i a : i\ge2, a\in \Sigma \}^\omega$ there are infinite square-free words over $\Sigma$ compatible with $\mu$.
Thus $d(\{0,1,2,3\})\le3$.\\

\item $\mathbf{7\le d(\{0,1,2\})\le 19}$:
Let $w=(0{\diamond}^51{\diamond}^52{\diamond}^5)^\omega$. 
An exhaustive search  confirms that there are only $4281$ square-free words over $\{0,1,2\}$ compatible with $w$.
Thus $d(\{0,1,2\})\ge7$.

Let $\Sigma = \{0,1,2\}$ and $W=\{\diamond^9\}\cup\{\diamond^i a: i\in\{18,\ldots, 26\}, a\in\Sigma\}.$
We can use Algorithm \ref{algosubgraph} to check that we can apply Lemma \ref{lemmagraphf} with $p=61$ and the values of $f$ given in Table \ref{tabfx}.
Thus conditions (I),(II) and (III) of Theorem \ref{mainTh} are fullfilled.
\begin{table}[htb]
\center
 \begin{tabular}{|c|c|c|c|c|c|c|c|c|c|c|}
 \hline
  $|w|$&9&19&20&21&22&23&24&25&26&27\rule{0pt}{12pt}\\[2px]\hline
  $f(|w|)$&4&19&22&28&36&50&63&88&118&148\rule{0pt}{12pt}\\[2px]\hline
  $x_{|w|}$&$\frac{27}{100}$  &$\frac{7}{100}$  &$\frac{13}{200}$  &$\frac{11}{200}$ &$\frac{9}{200}$ &$\frac{1}{25}$&$\frac{3}{100}$ &$\frac{1}{40}$ &$\frac{1}{40}$ &$\frac{1}{50}$\rule{0pt}{14pt}\\[4px] \hline 
 \end{tabular}
 \caption{The values of $f(|w|)$ and $x_{|w|}$ for the computation of $d(\{0,1,2\})$}
 \label{tabfx}
\end{table}
We can also check that the values of $x_{|w|}$ given in Table~\ref{tabfx} fulfill condition \eqref{condeq2} of Theorem \ref{mainTh}. 
We deduce that for any $\mu\in W^\omega$ there are infinite square-free words over $\Sigma$ compatible with $\mu$.
Moreover, $\{\diamond ^i a : i\ge18, a\in \Sigma \}^\omega\subseteq W^\omega.$
We deduce that for any $\mu\in \{\diamond ^i a : i\ge18, a\in \Sigma \}^\omega$ there are infinite square-free words over $\Sigma$ compatible with $\mu$.
Thus $d(\{0,1,2\})\le19$.
\end{proof}
The three applications of Algorithm \ref{algosubgraph} require between 30 and 100GB of RAM (and around 5 hours of computations). 
We had to optimize the way strings are stored in memory in order to be able to compute the graphs for large enough values of $p$.
The rest of the computations (finding the solution to the system and the exhaustive search) easily run
on a laptop in a few milliseconds.
Remark that we showed something slightly stronger since the results would still hold if an adversary was to tell us at every choice of letter only the next $5$ forced letters with their positions (that is, we know the next element of $W$).

Experimental computations suggest that $d(\{0,1,2\})$ is closer to 7 than to $19$ and that $d(\{0,1,2,3,4\})=2$.

\subsection*{Acknowledgement}
Computational resources have been provided by the Consortium des \'Equipements de Calcul Intensif (C\'ECI), 
funded by the Fonds de la Recherche Scientifique de Belgique (F.R.S.-FNRS) under Grant No. 2.5020.11 and by the Walloon Region


\begin{thebibliography}{1}
\bibitem{BELL20071295}
J.~P. Bell and T.~L. Goh.
\newblock Exponential lower bounds for the number of words of uniform length
  avoiding a pattern.
\newblock {\em Information and Computation}, 205(9):1295--1306, 2007.

\bibitem{mainquestionpandd}
J.~{Currie}, T.~{Harju}, P.~{Ochem}, and N.~{Rampersad}.
\newblock {Some further results on squarefree arithmetic progressions in
  infinite words}.
\newblock {\em Theoretical Computer Science}, 799:140--148, 2019.

\bibitem{CZERWINSKI2007453}
S.~Czerwiński and J.~Grytczuk.
\newblock Nonrepetitive colorings of graphs.
\newblock {\em Electronic Notes in Discrete Mathematics}, 28:453--459, 2007.
\newblock 6th Czech-Slovak International Symposium on Combinatorics, Graph
  Theory, Algorithms and Applications.

\bibitem{doi:10.1002/rsa.20411}
J.~Grytczuk, J.~Kozik, and P.~Micek.
\newblock New approach to nonrepetitive sequences.
\newblock {\em Random Structures \& Algorithms}, 42(2):214--225.

\bibitem{HARJU2018}
T.~Harju.
\newblock On square-free arithmetic progressions in infinite words.
\newblock {\em Theoretical Computer Science}, 2018.

\bibitem{Kolpakov2007}
R.~M. Kolpakov.
\newblock On the number of repetition-free words.
\newblock {\em Journal of Applied and Industrial Mathematics}, 1(4):453--462,
  2007.

\bibitem{doublepat}
P.~Ochem.
\newblock Doubled patterns are 3-avoidable.
\newblock {\em Electronic Journal of Combinatorics}, 23(1), 2016.

\bibitem{Thue06}
A.~Thue.
\newblock {\"{U}ber unendliche {Z}eichenreihen}.
\newblock {\em 'Norske Vid. Selsk. Skr. I. Mat. Nat. Kl. Christiania}, 7:1--22,
  1906.
\end{thebibliography}
\end{document}